\documentclass[aps,pra,reprint,nofootinbib,floatfix,longbibliography]{revtex4-2}
\usepackage[T1]{fontenc}
\usepackage[utf8]{inputenc}
\usepackage{microtype}
\usepackage{amsmath,amssymb,amsthm,mathtools}
\usepackage{xcolor}
\usepackage{graphicx}
\usepackage[colorlinks=true,linkcolor=blue!60!black,citecolor=blue!60!black,urlcolor=blue!60!black]{hyperref}

\newtheorem{theorem}{Theorem}[section]
\newtheorem{lemma}[theorem]{Lemma}
\newtheorem{proposition}[theorem]{Proposition}
\newtheorem{corollary}[theorem]{Corollary}

\newcommand{\C}{\mathbb C}
\newcommand{\R}{\mathbb R}
\newcommand{\A}{\mathcal A}
\newcommand{\Hh}{\mathcal H}
\newcommand{\Tr}{\operatorname{Tr}}
\newcommand{\diag}{\operatorname{diag}}
\newcommand{\norm}[1]{\left\lVert #1\right\rVert}
\newcommand{\abs}[1]{\left|#1\right|}
\newcommand{\ket}[1]{\lvert #1\rangle}
\newcommand{\bra}[1]{\langle #1\rvert}
\newcommand{\dist}{\mathrm d}

\begin{document}

\title{Calibrated Helstrom geometry on the Bloch ball via Connes spectral distance}
\author{Kaushlendra Kumar}
\affiliation{School of Mathematical Sciences, Queen Mary University of London, Mile End Road, London E1 4NS, United Kingdom\\
\href{mailto:kaushlendra.kumar@qmul.ac.uk}{kaushlendra.kumar@qmul.ac.uk}}

\begin{abstract}
{\bf Abstract}: We show that the equal-prior Helstrom trace-distance geometry of qubit states is recovered from Connes spectral distance in a finite scalar--qubit--scalar model. The two scalar reference sectors couple isotropically to the qubit block through identity Dirac links, so that the full Bloch ball, including mixed states, inherits its standard chordal trace-distance geometry from the finite spectral metric. The scalar-sector distances serve a distinct calibration role: they determine the individual link lengths, satisfy a Pythagorean consistency relation, and reconstruct the middle-sector scale.
\end{abstract}

\maketitle

\section{Introduction}\label{sec:introduction}

Trace distance is the natural operational metric for equal-prior binary discrimination of qubit states~\cite{Helstrom1976}. More explicitly, if one of two known states, \(\rho\) or \(\sigma\), is prepared with equal prior probability, the optimal success probability after optimizing over measurements is
\begin{equation}\label{eq:helstrom-success}
p_{\mathrm{succ}}^{\mathrm{opt}}(\rho,\sigma)=\frac12\bigl(1+T(\rho,\sigma)\bigr),
\quad
T(\rho,\sigma)=\frac12\norm{\rho-\sigma}_1
\end{equation}
where \(\norm{X}_1=\Tr\sqrt{X^\dagger X}\) is the trace norm. For Hermitian \(X\), this is the sum of the absolute values of its eigenvalues. The corresponding optimal error probability is then
\begin{equation}\label{eq:helstrom-error}
p_{\mathrm{err}}^{\mathrm{opt}}(\rho,\sigma)
=1-p_{\mathrm{succ}}^{\mathrm{opt}}(\rho,\sigma)
=\frac12\bigl(1-T(\rho,\sigma)\bigr).
\end{equation}
Here we ask how the same qubit trace geometry appears from Connes spectral distance in a finite noncommutative model. For unequal priors, the optimal success probability depends instead on the trace norm of the weighted Helstrom operator \(\lambda\rho-(1-\lambda)\sigma\), with \(\lambda\in[0,1]\). The present work is confined to the equal-prior case.

Connes spectral distance~\cite{Connes1994} is a variational metric built from spectral metric data \((\A,\Hh,\pi,D)\), where \(\A\) is an involutive algebra, \(\Hh\) is a Hilbert space, \(\pi:\A\to\mathcal B(\Hh)\) is a representation on the space of bounded operators on $\Hh$, and \(D\) is a Dirac operator. States are normalized positive linear functionals \(\varphi:\A\to\C\). For two states \(\varphi,\psi\), the distance is
\begin{equation}\label{eq:connes-distance}
{\rm d}_D(\varphi,\psi)=
\sup_{\substack{a=a^*\\ L_D(a)\le1}}
\abs{\varphi(a)-\psi(a)},
\end{equation}
where the self-adjoint test element \(a\) is varied across the Lipschitz ball, \(L_D(a)=\norm{[D,\pi(a)]}\). In finite dimensions, the formula becomes a concrete matrix-norm optimization.

The finite-space analysis of Iochum--Krajewski--Martinetti~\cite{IochumKrajewskiMartinetti2001} shows why the bare irreducible qubit algebra \(M_2(\C)\) does not give the full Bloch-ball trace metric: a two-eigenvalue Dirac operator controls only Pauli directions transverse to its eigenaxis, leaving state pairs along the missing direction at infinite distance. We illustrate this finite-slice obstruction in Bloch coordinates in Figure~\ref{fig:bare-M2-slice}, while recalling the calculation in Appendix~\ref{app:bare-M2}.

\begin{figure}[t]
\centering
\includegraphics[width=0.72\linewidth]{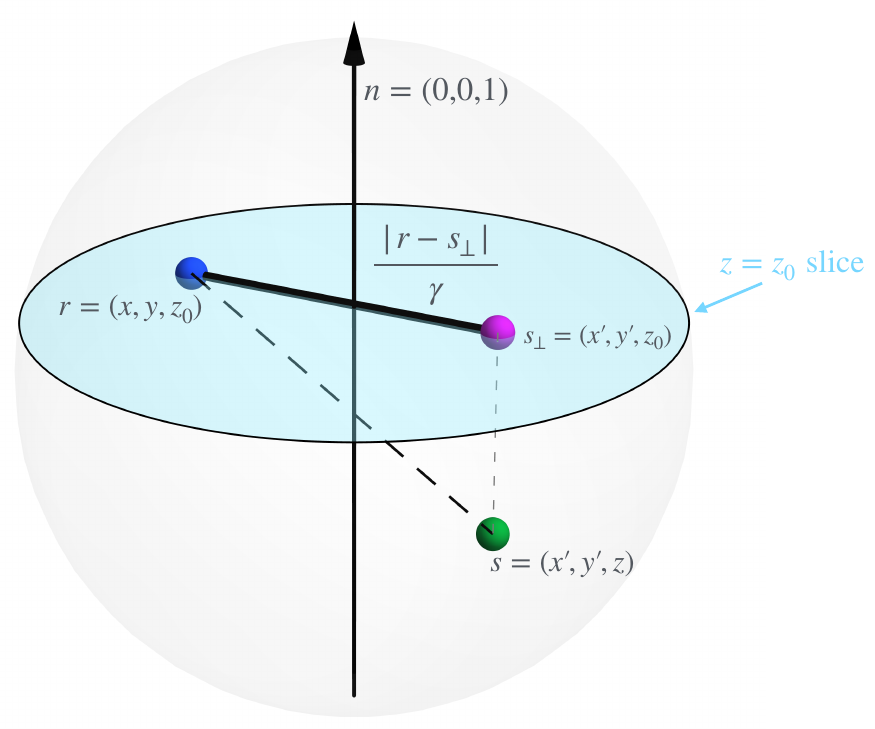}
\caption{Depiction of bare \(M_2(\C)\) finite-slice obstruction. In coordinates with Dirac axis \(n=(0,0,1)\), the bare two-eigenvalue Dirac operator controls only directions transverse to \(n\). Points in the same slice \(z=z_0\) have finite distance \(|r-s_\perp|/\gamma\), where \(\gamma=|d_+-d_-|\) is the Dirac eigenvalue gap. Any off-slice separation contains an unconstrained component parallel to \(n\), and the corresponding bare spectral distance is infinite.}
\label{fig:bare-M2-slice}
\end{figure}

The scalar-anchored model below removes this degeneracy by isotropically coupling the scalar sectors through identity blocks, so that all Pauli directions enter the Lipschitz ball. Other finite spectral data remove the obstruction in different ways. In the spin-\(1/2\) fuzzy sphere of D'Andrea--Lizzi--V\'arilly~\cite{DAndreaLizziVarilly2013} (in their \(N=1\) convention) with algebra \(M_2(\C)\), the exact qubit trace distance in Bloch coordinates $r,s\in B^3$ is recovered:
\begin{equation}\label{eq:fuzzy-spin-half-trace}
{\rm d}_{1/2}(\omega_{\rho(r)},\omega_{\rho(s)})=\frac12|r-s|,
\end{equation}
where the states on the Bloch ball \(B^3\) are \(\rho(r)=\frac12(I_2+r\cdot\sigma)\). Using the Hilbert--Schmidt operatorial formulation developed for the Moyal plane and fuzzy sphere~\cite{ScholtzChakraborty2013,DeviPrajapatMukhopadhyayChakrabortyScholtz2015}, Devi--Kumar--Chakraborty--Scholtz~\cite{KumarChakraborty2016} revisited these finite-distance calculations using Dirac eigenspinors; at \(n=1/2\) they recover the pure coherent-state sine law corresponding to Eq.~\eqref{eq:fuzzy-spin-half-trace} along with a linear interpolation for mixed states along a fixed chord. Recent finite-matrix qubit triples~\cite{LinXuWangChen2026,WangLinYou2026} also connect Connes distance with trace-distance-type quantities, coherence, discord, and unitary invariance, while phase-space models~\cite{LinHeng2020,LinHeng2022} exhibit additive and Pythagorean distance laws.

Here the route is different: the qubit block is anchored by two scalar reference sectors, each represented on a copy of \(\C^2\), the same two-component multiplicity as the qubit block. This representation choice, introduced in Sec.~\ref{sec:model}, is what allows the scalar--qubit couplings to be isotropic identity blocks rather than vector links selecting a preferred Bloch direction. The corresponding main result, proved in Sec.~\ref{sec:helstrom}, is
\begin{equation}\label{eq:intro-main-scale}
{\rm d}(\omega_\rho,\omega_\sigma)=\frac{2}{\Lambda}T(\rho,\sigma),
\qquad
\Lambda=\sqrt{\mu_L^2+\mu_R^2},
\end{equation}
which is the usual chordal trace-distance geometry of the full mixed-state Bloch ball, written in finite spectral units. Sec.~\ref{sec:calibration} then computes the scalar-sector distances, which determine the individual link lengths and satisfy
\[
{\rm d}(\omega_L,\omega_R)^2=\ell_L^2+\ell_R^2,
\qquad
\ell_L=\mu_L^{-1},\quad \ell_R=\mu_R^{-1},
\]
where \(\omega_L\) and \(\omega_R\) denote the two scalar-sector states. Finally, Sec.~\ref{sec:discussion} explains the relation between the one-scalar trace-norm mechanism and the two-reference calibration geometry; the one-scalar baseline itself is derived in Appendix~\ref{app:one-scalar}. Some supporting calculations for Secs.~\ref{sec:helstrom}--\ref{sec:calibration} are collected in Appendix~\ref{app:secondary}.

\section{The Isotropic Two-Anchor Model}\label{sec:model}

The spectral triple is \((\A,\Hh,\pi,D)\), with
\[
\begin{aligned}
\A&=\C_L\oplus M_2(\C)\oplus \C_R,\\
\pi(x,m,z)&=\diag(xI_2,m,zI_2),\\
\Hh&=\C^2\oplus \C^2\oplus \C^2,
\end{aligned}
\]
and the relevant states are
\[
\begin{aligned}
\omega_L(x,m,z)&=x,\\
\omega_\rho(x,m,z)&=\Tr(\rho m),\\
\omega_R(x,m,z)&=z,
\end{aligned}
\]
where \(\rho\) runs over qubit density matrices on \(\C^2\). We consider the isotropic Dirac operator
\[
D=
\begin{pmatrix}
0 & \mu_L I_2 & 0\\
\mu_L I_2 & 0 & \mu_R I_2\\
0 & \mu_R I_2 & 0
\end{pmatrix},
\qquad
\mu_L,\mu_R>0.
\]
Thus the scalar summands are represented on copies of \(\C^2\), and the off-diagonal Dirac links can be the identity blocks \(\mu_L I_2\) and \(\mu_R I_2\) between equal two-dimensional multiplicity spaces. No preferred Bloch direction is singled out. For \(a=(x,m,z)\), we write \(L_D(x,m,z):=L_D(a)\) and employ abbreviations
\[
\Delta_L:=m-xI_2,
\qquad
\Delta_R:=m-zI_2,
\]
such that for $a=(x,m,z)=a^*$,
\begin{equation}\label{eq:general-comm}
[D,\pi(a)] =
\begin{pmatrix}
0 & \mu_L\Delta_L & 0\\
-\mu_L\Delta_L & 0 & -\mu_R\Delta_R\\
0 & \mu_R\Delta_R & 0
\end{pmatrix}.
\end{equation}
The commutator depends only on $\Delta_L$ and $\Delta_R$, the differences between the middle observable and the two scalar values.

The norm in \(L_D\) is the operator norm, denoted \(\norm{\cdot}_{\mathrm{op}}\) when needed. We will routinely use the following $C^*$-algebra property for the operator norm in finite dimensions,
\begin{equation}\label{eq:cstar-norm}
\norm{A}_{\mathrm{op}}^2=\norm{A^*A}_{\mathrm{op}}
=\lambda_{\max}(A^*A),
\end{equation}
and, for block direct sums, we also have
\begin{equation}\label{eq:direct-sum-norm}
\norm{A_1\oplus A_2}_{\mathrm{op}}
=\max\{\norm{A_1}_{\mathrm{op}},\norm{A_2}_{\mathrm{op}}\}.
\end{equation}

\begin{lemma}\label{lem:comm-reduction}
If the eigenvalues of $m$ are $\lambda_1,\lambda_2$, then
\begin{equation}\label{eq:lipschitz}
L_D(x,m,z)^2
=
\max_{i=1,2}\Bigl(\mu_L^2(\lambda_i-x)^2+\mu_R^2(\lambda_i-z)^2\Bigr).
\end{equation}
\end{lemma}

\begin{proof}
The proof is a two-channel decomposition. First we diagonalize the represented operators using \(\widetilde U=U\oplus U\oplus U\) where \(U^*mU=\diag(\lambda_1,\lambda_2)\). Since the link blocks in \(D\) are \(\mu_L I_2\) and \(\mu_R I_2\), one has \(\widetilde U^*D\widetilde U=D\), while \(\widetilde U^*\pi(x,m,z)\widetilde U=\pi(x,U^*mU,z)\). The operator norm is invariant under this unitary conjugation, so it suffices to treat diagonal \(m\).

Next we reorder the eigenbasis to display the two channels. To that end, let \(L_i,M_i,R_i\) denote the $\C^6$-vectors in the left, middle, and right copies of \(\C^2\) such that the natural row and column order is \((L_1,L_2,M_1,M_2,R_1,R_2)\). We then reorder it as \((L_1,M_1,R_1,L_2,M_2,R_2)\), since permutation is an allowed unitary operation inside the operator norm, and the resulting commutator takes the direct-sum form \(T_1\oplus T_2\):
\[
T_i=
\begin{pmatrix}
0&\alpha_i&0\\
-\alpha_i&0&-\beta_i\\
0&\beta_i&0
\end{pmatrix},
\]
with
\[
\alpha_i=\mu_L(\lambda_i-x),
\qquad
\beta_i=\mu_R(\lambda_i-z).
\]
By \eqref{eq:direct-sum-norm}, it remains to compute one channel. A direct multiplication gives
\[
T_i^*T_i=
\begin{pmatrix}
\alpha_i^2&0&\alpha_i\beta_i\\
0&\alpha_i^2+\beta_i^2&0\\
\alpha_i\beta_i&0&\beta_i^2
\end{pmatrix},
\]
whose eigenvalues are
\[
0,\qquad
\alpha_i^2+\beta_i^2,\qquad
\alpha_i^2+\beta_i^2.
\]
By Eq.~\eqref{eq:cstar-norm},
\[
\norm{T_i}_{\mathrm{op}}^2=\alpha_i^2+\beta_i^2
=\mu_L^2(\lambda_i-x)^2+\mu_R^2(\lambda_i-z)^2.
\]
Taking the larger of the two channels proves \eqref{eq:lipschitz}.
\end{proof}

\section{Exact Helstrom Geometry on the Bloch Ball}\label{sec:helstrom}

The key metric statement of this paper is the following exact identification on the full qubit state space.

\begin{theorem}[Calibrated trace distance]\label{thm:mixed}
For any two middle-sector qubit states $\omega_\rho,\omega_\sigma$,
\begin{equation}\label{eq:mixed-main}
{\rm d}(\omega_\rho,\omega_\sigma)
=
\frac{2}{\Lambda}\,T(\rho,\sigma),
\qquad
\Lambda=(\mu_L^2+\mu_R^2)^{1/2}.
\end{equation}
\end{theorem}

\begin{proof}
If \(\rho=\sigma\), both sides vanish. Assume henceforth that \(\Delta:=\rho-\sigma\ne0\). The middle-state difference is
\[
\omega_\rho(a)-\omega_\sigma(a)=\Tr(\Delta m),
\]
and it is independent of the scalar entries \(x,z\). We first eliminate the scalar variables for a fixed middle observable, then optimize the remaining qubit direction.

\paragraph*{Step 1: scalar-anchor reduction.}
For fixed \(m\), the scalar variables \(x,z\) do not affect the middle-state difference; they only determine the Lipschitz seminorm of the completed algebra element. Thus the relevant question is whether the fiber over \(m\) intersects the Lipschitz ball:
\[
\exists\,x,z:\ L_D(x,m,z)\le1 .
\]
This is equivalent to \(C(m)\le1\), where
\[
C(m):=\inf_{x,z}L_D(x,m,z).
\]
Thus taking the infimum over \(x,z\) does not restrict the variational problem; it tests whether the fiber over \(m\) meets the Lipschitz ball. Next, we write \(m=m^*\in M_2(\C)\) in midpoint--gap form using spectral decomposition:
\begin{equation}\label{eq:middle-param}
\begin{gathered}
\bar\lambda_m:=\frac{\lambda_++\lambda_-}{2},
\qquad
\gamma_m:=\lambda_+-\lambda_-,\\
m=\bar\lambda_m I_2+\frac{\gamma_m}{2}(2P-I_2),
\end{gathered}
\end{equation}
where \(P\) is the rank-one projection onto the largest eigensector, \(\lambda_+\ge\lambda_-\), and determines the qubit direction. Since \(\Tr(\Delta)=0\), we get
\[
\omega_\rho(m)-\omega_\sigma(m)=\Tr(\Delta m)=\gamma_m\,\Tr(\Delta P).
\]
By Lemma~\ref{lem:comm-reduction}, the two squared channel terms are
\[
Q_\pm(x,z)=
\mu_L^2\left(\bar\lambda_m\pm\frac{\gamma_m}{2}-x\right)^2+
\mu_R^2\left(\bar\lambda_m\pm\frac{\gamma_m}{2}-z\right)^2.
\]
Thus \(L_D(x,m,z)^2=\max\{Q_+(x,z),Q_-(x,z)\}\).
Writing \(u=x-\bar\lambda_m\) and \(v=z-\bar\lambda_m\), one has the chain
\[
\begin{aligned}
\max\{Q_+,Q_-\}
&\ge \frac{Q_++Q_-}{2} \\
&=\frac{\Lambda^2\gamma_m^2}{4}+\mu_L^2u^2+\mu_R^2v^2 \\
&\ge \frac{\Lambda^2\gamma_m^2}{4}.
\end{aligned}
\]
Equality holds at \(u=v=0\), i.e. at \(x=z=\bar\lambda_m\), where the two channel terms coincide. Therefore, the infimum is a minimum and we have
\[
\begin{aligned}
C(m)^2
&=\inf_{x,z}\max\{Q_+(x,z),Q_-(x,z)\}\\
&=\frac{\Lambda^2\gamma_m^2}{4},
\qquad
C(m)=\frac{\Lambda\gamma_m}{2}.
\end{aligned}
\]
Note that \(C(m)\) depends only on the eigenvalue gap \(\gamma_m\), not on the midpoint \(\bar\lambda_m\) or on the projection direction \(P\). Since the minimum is attained at \(x=z=\bar\lambda_m\), the fiber criterion reduces the original variational problem to
\begin{align}\label{eq:middle-variational-reduction}
\dist(\omega_\rho,\omega_\sigma)
&=
\sup_{\substack{m=m^*,\,x,z\\
L_D(x,m,z)\le1}}
\abs{\Tr(\Delta m)}
\nonumber\\
&=
\sup_{\substack{m=m^*\\ C(m)\le1}}
\abs{\Tr(\Delta m)}
\nonumber\\
&=
\sup_{\substack{\gamma_m\ge0,\;P\\ \Lambda\gamma_m/2\le1}}
\gamma_m\,\abs{\Tr(\Delta P)}\\
&=
\frac{2}{\Lambda}\sup_P\abs{\Tr(\Delta P)}.
\nonumber
\end{align}

\paragraph*{Step 2: qubit-direction optimization.}
Since \(\Delta\) is traceless and Hermitian on a qubit, its eigenvalues are \(\pm\eta\) for some \(\eta>0\). Let \(e_\pm\) be corresponding unit eigenvectors, and let \(P=\ket{\psi}\bra{\psi}\) with \(\psi=a e_+ + b e_-\), \(|a|^2+|b|^2=1\). Then
\[
\Tr(\Delta P)=\langle\psi,\Delta\psi\rangle
=\eta(|a|^2-|b|^2),
\]
so \(\abs{\Tr(\Delta P)}\le\eta\). Equality is attained at the positive or negative spectral projection of \(\Delta\). Therefore using  \(\norm{\Delta}_1=|\eta|+|-\eta|=2\eta\) for this Hermitian \(\Delta\) we get
\[
\sup_P\abs{\Tr(\Delta P)}=\eta=\frac12\norm{\Delta}_1.
\]
Putting the two optimization steps together, and using \(T(\rho,\sigma)=\norm{\Delta}_1/2\), gives the upper bound in \eqref{eq:mixed-main}.

\paragraph*{Attainment.}
Let \(P_+\) be the positive spectral projection of \(\Delta\), set \(\gamma_m=2/\Lambda\), and use the minimizing scalar extension \(x=z=\bar\lambda_m\). Since the midpoint does not affect \(\Tr(\Delta m)\), we can set \(\bar\lambda_m=0\) without loss of generality and write
\[
m=\frac1{\Lambda}(2P_+-I_2),
\qquad x=z=0.
\]
For this element, the two eigenvalue channels have costs
\[
Q_+=Q_-=\Lambda^2\left(\frac1{\Lambda}\right)^2=1,
\]
so \(L_D(x,m,z)=1\). Moreover,
\[
\abs{\Tr(\Delta m)}
=\frac{2}{\Lambda}\abs{\Tr(\Delta P_+)}
=\frac1{\Lambda}\norm{\Delta}_1
=\frac{2}{\Lambda}T(\rho,\sigma).
\]
Thus the upper bound is attained, proving equality.
\end{proof}

\paragraph*{Calibrated discrimination probabilities.}
The spectral scale \(2/\Lambda\) rescales all middle-sector distances uniformly, while the inverse factor \(\Lambda/2\) converts the raw spectral distance back into trace distance. Writing \(\dist_{\rho\sigma}:=\dist(\omega_\rho,\omega_\sigma)\), the theorem gives the calibrated readout
\begin{equation}\label{eq:calibrated-readout}
T(\rho,\sigma)=\frac{\Lambda}{2}\,\dist_{\rho\sigma},
\quad
p_{\mathrm{succ}}^{\mathrm{opt}}(\rho,\sigma)
=\frac12\left(1+\frac{\Lambda}{2}\dist_{\rho\sigma}\right),
\end{equation}
where the success-probability formula is just Eq.~\eqref{eq:helstrom-success} after calibration. The error probability follows from Eq.~\eqref{eq:helstrom-error}. Thus increasing the couplings \(\mu_L,\mu_R\) shrinks spectral lengths, while the calibrated combination \(\Lambda \dist/2\) is exactly the trace distance controlling the discrimination task. No purity assumption is made: the identification is global on the full mixed-state Bloch ball.

\begin{corollary}[Bloch-ball form]\label{cor:bloch-ball}
For Bloch ball states,
\[
\rho(r)=\frac12(I+r\cdot \sigma),
\quad
\rho(s)=\frac12(I+s\cdot \sigma),\quad
|r|,|s|\le 1,
\]
the calibrated distance is~\cite{NielsenChuang2010}
\begin{equation}\label{eq:bloch-ball}
\dist_{rs}:=\dist(\omega_{\rho(r)},\omega_{\rho(s)})=\frac{|r-s|}{\Lambda}.
\end{equation}
\end{corollary}
This is easily seen as follows. Set \(a=r-s\) and use \((a\cdot\sigma)^2=|a|^2I\), so that the operator \(\rho(r)-\rho(s)=\frac12a\cdot\sigma\) has eigenvalues \(\pm |r-s|/2\), yielding
\[
\norm{\rho(r)-\rho(s)}_1
=\left|\frac{|r-s|}{2}\right|
+\left|-\frac{|r-s|}{2}\right|
=|r-s|.
\]
Theorem~\ref{thm:mixed} then produces \eqref{eq:bloch-ball} and thus the calibrated success probability in this case is
\begin{equation}\label{eq:bloch-success}
p_{\mathrm{succ}}^{\mathrm{opt}}(\rho(r),\rho(s))
=\frac12\left(1+\frac{|r-s|}{2}\right)
=\frac12\left(1+\frac{\Lambda}{2}\dist_{rs}\right).
\end{equation}
The middle-sector chord \(\Lambda^{-1}|r-s|\) is displayed together with the scalar calibration lengths in Fig.~\ref{fig:scalar-calibration}.

\section{Two-Reference Calibration Geometry}\label{sec:calibration}

The same finite data also determine a two-reference scalar calibration geometry around the qubit block. We now record the pure-state boundary and the distances involving the two outer scalar sectors.

\begin{proposition}\label{prop:secondary}
In the isotropic model one has
\begin{align}
\dist(\omega_L,\omega_R)&=\sqrt{\frac1{\mu_L^2}+\frac1{\mu_R^2}}, \label{eq:LR}\\
\dist(\omega_L,\omega_\rho)&=\frac1{\mu_L},
\qquad
\dist(\omega_R,\omega_\rho)=\frac1{\mu_R}, \label{eq:scalar-middle}\\
\dist(\omega_v,\omega_w)&=\frac{2}{\Lambda}\sin\frac{\Theta(v,w)}2, \label{eq:vw}
\end{align}
for every qubit state $\omega_\rho$ and all pure middle states \(\omega_v:=\omega_{\ket v\bra v}\), \(\omega_w:=\omega_{\ket w\bra w}\), where \(v,w\in\C^2\) are unit vectors and \(\Theta(v,w)\) is the angle between their Bloch vectors.
\end{proposition}

The derivation is collected in Appendix~\ref{app:secondary}. Equation~\eqref{eq:scalar-middle} shows that each scalar sector sits at a fixed spectral radius from the entire qubit sector, independently of \(\rho\). The optimal middle observable can be chosen proportional to \(I_2\), so the outer sectors see the qubit block as a whole, while qubit distinguishability is carried by middle--middle distances.

Equation~\eqref{eq:vw} is the pure-state restriction of the Bloch-ball theorem. If \(P_v=\ket v\bra v=\rho(r_v)\) and \(P_w=\ket w\bra w=\rho(r_w)\), with \(r_v,r_w\in S^2\), then
\[
|r_v-r_w|=2\sin\frac{\Theta(v,w)}2,
\]
so the spectral metric on pure states is the Euclidean chordal metric on the Bloch sphere, not the geodesic arc. This is the same sine-law structure that appeared at the spin-\(1/2\) fuzzy-sphere endpoint in Devi \textit{et al.}~\cite{KumarChakraborty2016},
\[
\dist_{1/2}(\omega_\theta,\omega_0)=r_{1/2}\sin\frac{\theta}{2},
\]
where \(\omega_0\) and \(\omega_\theta\) are pure coherent states separated by Bloch angle \(\theta\). The coefficient \(r_{1/2}\) is the fuzzy-sphere radius normalization; in Eq.~\eqref{eq:vw} it is replaced by \(2/\Lambda\), fixed by the scalar links. The mixed states discussed there lie on a selected chord, whereas Theorem~\ref{thm:mixed} gives the arbitrary-pair formula on the whole Bloch ball.

The scalar--scalar formula should be read differently. The scalar states are not points of the Bloch ball. They live on different central summands, so the scalar-sector lengths are not Helstrom probabilities and are not distances from the center \(r=0\). Their role is to record the finite link lengths by which the scalar references see the qubit block as a whole.

This gives a direct calibration reading. If
\[
\ell_L:=\dist(\omega_L,\omega_\rho)=\frac1{\mu_L},
\qquad
\ell_R:=\dist(\omega_R,\omega_\rho)=\frac1{\mu_R},
\]
then the combined scale appearing in the middle-sector Helstrom formula is fixed by the scalar-sector measurements:
\[
\Lambda=\sqrt{\ell_L^{-2}+\ell_R^{-2}}.
\]
The scalar--scalar distance provides an independent consistency relation,
\[
\dist(\omega_L,\omega_R)^2=\ell_L^2+\ell_R^2.
\]
Thus the scalar sectors reconstruct the combined middle-sector scale, resolve the individual link strengths hidden from middle--middle distances, and provide a Pythagorean check. Middle--middle distances see only \(\Lambda\); the radii \(\ell_L,\ell_R\) resolve the split, and the scalar--scalar distance checks their quadratic combination.

The two roles of the geometry are summarized in Fig.~\ref{fig:scalar-calibration}: the qubit block carries the calibrated trace/chord metric, while the scalar sectors record state-independent link lengths and the Pythagorean scalar-sector consistency relation.

\begin{center}
\refstepcounter{figure}\label{fig:scalar-calibration}
\centering
\includegraphics[width=0.86\columnwidth]{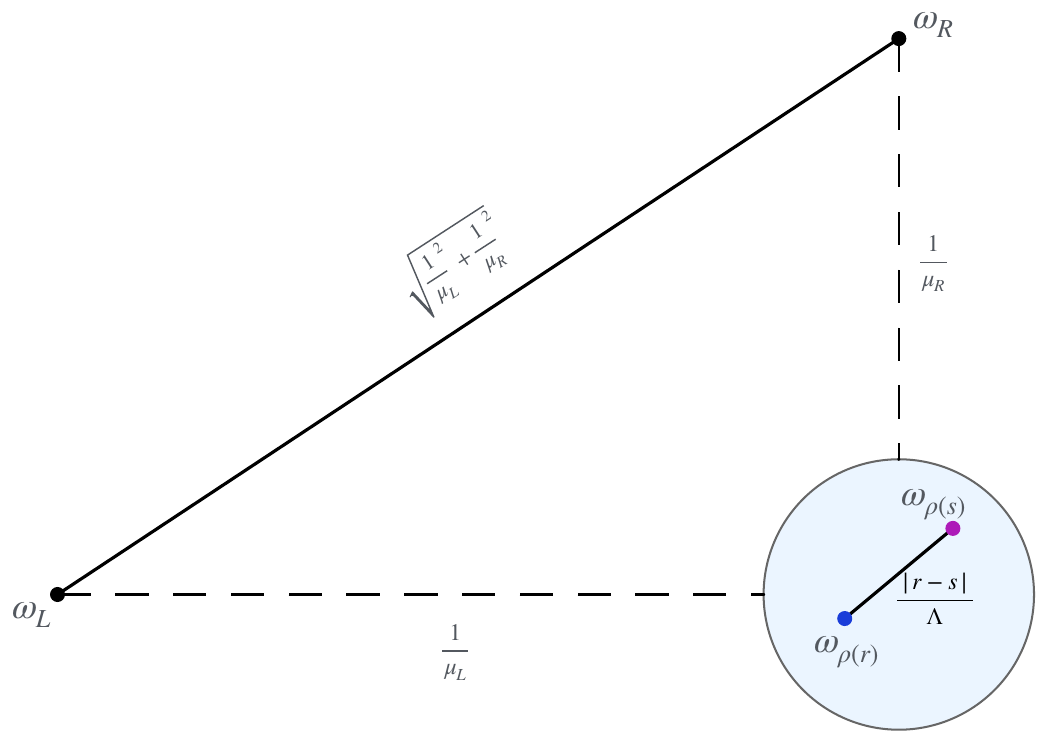}
\par\smallskip
\begin{minipage}{0.96\columnwidth}
\small \textbf{FIG.~\thefigure.} Scalar-anchored calibration geometry. The scalar states \(\omega_L,\omega_R\) define link lengths \(\ell_L=\mu_L^{-1}\), \(\ell_R=\mu_R^{-1}\), with \(\dist(\omega_L,\omega_R)^2=\ell_L^2+\ell_R^2\). Inside the qubit block, the middle-sector distance is the calibrated chord \(\dist(\omega_{\rho(r)},\omega_{\rho(s)})=\Lambda^{-1}|r-s|\).
\end{minipage}
\end{center}

\paragraph{Scalar-route consistency.}
The dashed route in Fig.~\ref{fig:scalar-calibration} should not be interpreted as a preferred path through a fixed middle state. Equation~\eqref{eq:scalar-middle} holds for every \(\omega_\rho\): each scalar sector is at a fixed spectral radius from the entire qubit block. For any choice of middle state,
\[
\dist(\omega_L,\omega_R)
<
\dist(\omega_L,\omega_\rho)+\dist(\omega_\rho,\omega_R)
=\frac1{\mu_L}+\frac1{\mu_R}.
\]
Thus the direct scalar--scalar distance is strictly shorter than the broken scalar--middle--scalar reading, giving a consistency check on the two scalar links.

\section{Discussion and Outlook}\label{sec:discussion}

The result is an exact scalar-reference spectral realization of the standard qubit trace-distance geometry on the full Bloch ball. The identity
\[
\dist(\omega_\rho,\omega_\sigma)=\frac{2}{\Lambda}T(\rho,\sigma)
\]
turns Connes distance into a calibrated Helstrom distance: the spectral length \(\dist\) recovers \(T(\rho,\sigma)\) after multiplication by \(\Lambda/2\). Since trace distance is a basic benchmark in discrimination, tomography, state verification, and direct estimation \cite{ChenEtAl2022}, the point is not to replace that metric, but to give an exact finite-spectral route to it.

The trace-norm mechanism itself is already present in the one-scalar identity-block baseline of Appendix~\ref{app:one-scalar}. A single scalar anchor removes the bare \(M_2(\C)\) finite-slice degeneracy by making all Pauli directions enter the ball condition:
\[
\dist_\mu(\omega_\rho,\omega_\sigma)=\frac2\mu T(\rho,\sigma).
\]
The second scalar reference upgrades this mechanism to a two-reference calibration geometry. As summarized in Fig.~\ref{fig:scalar-calibration}, the middle sector sees only the combined scale \(\Lambda\), while scalar-sector distances resolve the individual link lengths and verify the Pythagorean consistency relation \(\dist(\omega_L,\omega_R)^2=\ell_L^2+\ell_R^2\).

The same viewpoint points to several extensions. Time-varying scalar links introduce a moving calibration scale, so spectral lengths must be corrected before they are read as trace-distance information flow. Anisotropic identity-link deformations would separate an overall Helstrom scale from axis-sensitive distortions. The Pythagorean scalar-sector identity also suggests a calibration reading of similar distance laws in doubled and phase-space spectral geometries \cite{LinHeng2020,LinHeng2022,KumarChakraborty2018}. A separate future question is whether the scalar-reference calibration mechanism can be engineered, or at least simulated, in a concrete laboratory setting.

\begin{acknowledgments}
The author thanks Biswajit Chakraborty for previous mentorship and related past collaborations, and acknowledges generous support from the DFG through a Walter Benjamin Fellowship, grant no. 515782239.
\end{acknowledgments}

\appendix

\section{\texorpdfstring{Bare \(M_2(\C)\) Finite-Slice Calculation}{Bare M2(C) Finite-Slice Calculation}}\label{app:bare-M2}

For completeness, we recall the elementary obstruction behind the bare irreducible \(M_2(\C)\) example of Iochum et al.~\cite{IochumKrajewskiMartinetti2001}. In this appendix \(\dist\) denotes the spectral distance induced by the bare Dirac operator. Let \(D_0\) be a non-scalar Hermitian \(2\times2\) Dirac operator. Using the same midpoint--gap convention as in Eq.~\eqref{eq:middle-param}, write
\[
\begin{gathered}
D_0=\bar d\,I_2+\frac{\gamma}{2}\,n\cdot\sigma,
\qquad |n|=1,\\
\bar d=\frac{d_++d_-}{2},
\qquad
\gamma=|d_+-d_-|>0,
\end{gathered}
\]
where \(d_\pm\) are the two eigenvalues of \(D_0\).
For a self-adjoint algebra element
\[
a=a_0I_2+h\cdot\sigma,\qquad h\in\R^3,
\]
the commutator \([D_0,a]\) simplifies via the Pauli identity,
\[
[n\cdot\sigma,h\cdot\sigma]=2i(n\times h)\cdot\sigma,
\]
producing the following Lipschitz seminorm (the operator norm for $(n\times h)\cdot\sigma$ is simply $|n\times h|$):
\begin{equation}\label{eq:bare-m2-lip}
L_{D_0}(a)=\norm{[D_0,a]}=\gamma |n\times h|.
\end{equation}
Thus the ball condition \(L_{D_0}(a)\le1\) constrains only the component of \(h\) perpendicular to the Dirac axis \(n\).

For Bloch states \(\omega_r,\omega_s\) on \(M_2(\C)\),
\[
\omega_r(a)-\omega_s(a)=(r-s)\cdot h.
\]
If \((r-s)\cdot n\ne0\), the choice \(h=tn\) has \(L_{D_0}(a)=0\) for all \(t\) and makes the state separation unbounded. If \((r-s)\cdot n=0\), then the parallel component of \(h\) does not contribute to the state difference, and Cauchy--Schwarz gives
\[
\abs{(r-s)\cdot h}\le |r-s|\,|h_\perp|\le \frac{|r-s|}{\gamma}.
\]
The upper bound is attained by the optimal algebra element with \(h=(r-s)/(\gamma |r-s|)\). Hence the bare \(M_2(\C)\) distance is
\begin{equation}\label{eq:bare-m2-slice}
\dist(\omega_r,\omega_s)=
\begin{cases}
|r-s|/\gamma, & (r-s)\cdot n=0,\\[0.6ex]
+\infty, & (r-s)\cdot n\ne 0.
\end{cases}
\end{equation}

\section{Pure-State and Scalar-Sector Derivations}\label{app:secondary}

\begin{proof}[Proof of Proposition~\ref{prop:secondary}]
For \(\dist(\omega_L,\omega_R)\), the proof is complementary to the middle--middle reduction in Theorem~\ref{thm:mixed}. The objective is now \(|x-z|\), so the scalar variables \(x,z\) enter the quantity being maximized, while the middle observable \(m\) only controls the Lipschitz cost. Thus we define
\[
C_{LR}(x,z):=\inf_{m=m^*}L_D(x,m,z).
\]
The same fiber criterion gives
\begin{align*}
\dist(\omega_L,\omega_R)
&=
\sup_{\substack{x,z,\,m=m^*\\
L_D(x,m,z)\le1}}
\abs{x-z}\\
&=
\sup_{\substack{x,z\\ C_{LR}(x,z)\le1}}
\abs{x-z}.
\end{align*}
Here the infimum is the test for whether the fixed-\((x,z)\) fiber meets the Lipschitz ball; the computation below shows that it is attained.

To compute \(C_{LR}\), set
\[
Q(\lambda):=\mu_L^2(\lambda-x)^2+\mu_R^2(\lambda-z)^2 .
\]
By Lemma~\ref{lem:comm-reduction}, admissibility of a chosen \(m\) means
\[
Q(\lambda_i)\le 1
\]
for each eigenvalue \(\lambda_i\) of \(m\). For fixed \(x,z\), the function \(Q\) is minimized at
\[
\lambda_*=\frac{\mu_L^2x+\mu_R^2z}{\Lambda^2},
\]
with minimum
\[
\frac{\mu_L^2\mu_R^2}{\Lambda^2}(x-z)^2.
\]
Thus \(Q(\lambda_i)\ge Q(\lambda_*)\) for every eigenvalue channel, and the bound is sharp by taking \(m=\lambda_*I_2\). Hence
\begin{equation}\label{eq:CLR-cost}
C_{LR}(x,z)=\frac{\mu_L\mu_R}{\Lambda}\abs{x-z}.
\end{equation}
The condition \(C_{LR}(x,z)\le1\) is therefore equivalent to
\[
\abs{x-z}\le \sqrt{\frac1{\mu_L^2}+\frac1{\mu_R^2}},
\]
and equality is attained by choosing \(|x-z|\) at the bound. This proves \eqref{eq:LR}.

For \eqref{eq:scalar-middle}, the variational problem is
\[
\dist(\omega_L,\omega_\rho)=
\sup_{\substack{x,z,\,m=m^*\\
L_D(x,m,z)\le1}}
\abs{x-\Tr(\rho m)}.
\]
Let $\rho$ be any qubit density matrix. If \(m=\lambda_1P_1+\lambda_2P_2\), positivity gives \(p_i:=\Tr(\rho P_i)\ge0\) and \(p_1+p_2=1\). Thus, by the usual convexity bound for state expectations,
\[
\omega_\rho(m)=p_1\lambda_1+p_2\lambda_2
\]
lies in the convex hull of the two eigenvalues. Hence
\[
\begin{aligned}
\abs{x-\omega_\rho(m)}
&=\abs{p_1(x-\lambda_1)+p_2(x-\lambda_2)}\\
&\le p_1\abs{x-\lambda_1}+p_2\abs{x-\lambda_2}\\
&\le \max_i\abs{x-\lambda_i}.
\end{aligned}
\]
The first inequality is the triangle inequality, and the second uses \(p_i\ge0\) with \(p_1+p_2=1\); equivalently this is the spectral form of the usual state/operator-norm bound.
By Lemma~\ref{lem:comm-reduction}, admissibility gives, for each \(i\),
\[
\mu_L^2(\lambda_i-x)^2
\le
\mu_L^2(\lambda_i-x)^2+\mu_R^2(\lambda_i-z)^2
\le 1.
\]
Thus \(\abs{x-\lambda_i}\le1/\mu_L\), and therefore \(\dist(\omega_L,\omega_\rho)\le1/\mu_L\). Equality is attained by taking \(m=\lambda I_2\), \(z=\lambda\), and \(x=\lambda+1/\mu_L\). Then \(L_D=1\) and \(\abs{x-\Tr(\rho m)}=1/\mu_L\). The formula for \(\dist(\omega_R,\omega_\rho)\) is symmetric.

For \eqref{eq:vw}, write \(P_v=\ket v\bra v=\frac12(I+r_v\cdot\sigma)\) and \(P_w=\ket w\bra w=\frac12(I+r_w\cdot\sigma)\), with \(r_v,r_w\in S^2\). Then
\[
P_v-P_w=\frac12(r_v-r_w)\cdot\sigma,
\]
so \(P_v-P_w\) has eigenvalues
\[
\pm\frac12|r_v-r_w|=\pm\sin\frac{\Theta(v,w)}2.
\]
Therefore \(\norm{P_v-P_w}_1=2\sin(\Theta(v,w)/2)\). Applying Theorem~\ref{thm:mixed} to the pure states \(P_v\) and \(P_w\) gives \eqref{eq:vw}.
\end{proof}

\section{One-Scalar Identity-Block Baseline}\label{app:one-scalar}

The middle-sector trace-norm mechanism does not require two scalar references. In the one-anchor baseline, \(\C\oplus M_2(\C)\) is represented on \(\C^2\oplus\C^2\) by
\[
\begin{gathered}
\pi(x,m)=\diag(xI_2,m),\\
\omega(x,m)=x,\qquad
\omega_\tau(x,m)=\Tr(\tau m)\quad(\tau=\rho,\sigma),\\
D_\mu=
\begin{pmatrix}
0&\mu I_2\\
\mu I_2&0
\end{pmatrix},
\qquad \mu>0 .
\end{gathered}
\]
The shared twofold multiplicity allows the scalar--qubit link to be the isotropic identity block \(\mu I_2\). For \(a=(x,m)=a^*\),
\[
[D_\mu,\pi(x,m)]
=
\begin{pmatrix}
0&\mu(m-xI_2)\\
-\mu(m-xI_2)&0
\end{pmatrix},
\]
hence
\begin{equation}\label{eq:one-scalar-lip}
L_\mu(x,m):=\norm{[D_\mu,\pi(x,m)]}
=\mu\norm{m-xI_2}_{\mathrm{op}} .
\end{equation}
Using the midpoint--gap notation of Eq.~\eqref{eq:middle-param}, the minimizing scalar value is \(x=\bar\lambda_m\), and therefore
\[
\begin{aligned}
C_\mu(m)&:=\inf_x L_\mu(x,m)=\frac{\mu\gamma_m}{2},\\
C_\mu(m)\le1&\Longleftrightarrow \gamma_m\le\frac2\mu .
\end{aligned}
\]
For \(\Delta=\rho-\sigma\), the same projection optimization as in Theorem~\ref{thm:mixed} gives
\[
\begin{aligned}
\dist_\mu(\omega_\rho,\omega_\sigma)
&=
\sup_{\substack{\gamma_m\ge0,\;P\\ \mu\gamma_m/2\le1}}
\gamma_m\abs{\Tr(\Delta P)}\\
&=\frac2\mu\sup_P\abs{\Tr(\Delta P)}
=\frac1\mu\norm{\Delta}_1.
\end{aligned}
\]
Therefore, for all qubit density matrices,
\begin{equation}\label{eq:one-scalar-middle}
\dist_\mu(\omega_\rho,\omega_\sigma)
=\frac2\mu T(\rho,\sigma).
\end{equation}
The scalar-to-middle distance is also fixed:
\begin{equation}\label{eq:one-scalar-link}
\dist_\mu(\omega,\omega_\rho)=\frac1\mu .
\end{equation}
Indeed, Eq.~\eqref{eq:one-scalar-lip} implies \(\abs{x-\omega_\rho(m)}\le 1/\mu\), and equality is attained by taking \(m=\lambda I_2\), \(x=\lambda+1/\mu\).

\end{document}